\RequirePackage[l2tabu, orthodox]{nag}
\RequirePackage{snapshot}

\documentclass[10pt,onecolumn]{extarticle}

\makeatletter
\if@twocolumn
  \usepackage[dvips,letterpaper,top=0.5in, bottom=0.5in, left=0.75in, right=0.5in,includefoot,heightrounded]{geometry}
\else
  \usepackage[dvips,letterpaper,margin=1in,includefoot,heightrounded]{geometry}
\fi

\usepackage{srcltx}

\usepackage[russian,portuges,english]{babel}

\iflanguage{portuges}
    {\newcommand{\keywordname}{Palavras-chaves}}
    {\newcommand{\keywordname}{Keywords}}

\usepackage{amsmath}
\usepackage{amssymb,amsfonts}

\usepackage{soul} %

\usepackage{abstract}

\usepackage{graphicx}
\usepackage[usenames,dvipsnames,svgnames,x11names]{xcolor}

\usepackage{booktabs}

\usepackage{setspace}
\usepackage{flushend}

\usepackage{cite}

\usepackage{hyperref}
\usepackage[normalem]{ulem}

\usepackage{enumerate}

\usepackage[noend]{algpseudocode}

\usepackage{listings}

\usepackage[short,12hr]{datetime}
       \usepackage{fouriernc}

\usepackage{amsthm}   %

\newtheorem{theorem}{Theorem}[section]
\newtheorem{lemma}[theorem]{Lemma}
\newtheorem{corollary}[theorem]{Corollary}
\newtheorem{remark}[theorem]{Remark}

\newtheorem{algorithm}[theorem]{Algorithm}

\newcommand{\be}{\begin{equation}}
\newcommand{\ee}{\end{equation}}

\newcommand{\bD}{{\bf D}}

\newcommand{\bP}{{\bf P}}
\newcommand{\bI}{{\bf I}}

\newcommand{\bC}{{\bf C}}
\newcommand{\bA}{{\bf A}}

\newcommand{\bfe}{{\bf e}}

\makeatletter

\newcommand{\printtitle}{%
\makeatletter
\if@twocolumn

\twocolumn[%
  \maketitle
  \begin{onecolabstract}
    \myabstract
  \end{onecolabstract}
  \begin{center}
    \small
    \textbf{\keywordname}
    \\\medskip
    \mykeywords
  \end{center}
  \bigskip
]
\saythanks
\else
  \maketitle
  \begin{onecolabstract}
    \myabstract
  \end{onecolabstract}
  \begin{center}
    \small
    \textbf{\keywordname}
    \\\medskip
    \mykeywords
  \end{center}
  \bigskip
  \onehalfspacing
\fi
\makeatother
}

\author{%
S. M. Perera%
\thanks{S. M. Perera is with the Department of Mathematics, Embry-Riddle Aeronautical University, Daytona Beach, Florida 32114, USA (e-mail: pereras2@erau.edu).}
\and
A. Madanayake%
\thanks{A. Madanayake is with Department of Electrical and Computer Engineering, Florida International University, Miami, FL 33174, USA (e-mail: amadanay@fiu.edu).}
\and
R. J. Cintra%
\thanks{%
R. J. Cintra is with the
Signal Processing Group,
UFPE, Brazil
(e-mail: \url{rjdsc@de.ufpe.br}).}
}

\title{%
Efficient and Self-Recursive Delay Vandermonde Algorithm for Multi-Beam Antenna Arrays}

\newcommand{\myabstract}{%
This paper presents a self-contained factorization for the delay
Vandermonde matrix (DVM), which is the super class of the discrete
Fourier transform, using sparse and companion matrices. An efficient
DVM algorithm is proposed to reduce the complexity
of radio-frequency (RF) $N$-beam analog beamforming systems. There exist applications for wideband multi-beam beamformers
in wireless communication networks such as 5G/6G systems, system capacity can be improved
by exploiting the improvement of the signal to noise ratio (SNR) using coherent summation of propagating waves based on their directions of propagation. The presence of a multitude of RF beams allows multiple independent wireless links to be established at high SNR, or used
in conjunction with multiple-input multiple-output (MIMO) wireless systems, with the overall goal of improving system SNR and therefore capacity.
To realize such multi-beam beamformers at  acceptable analog circuit complexities, we use sparse factorization of the
DVM in order to derive a low arithmetic complexity DVM algorithm. The paper
also establishes an error bound and stability analysis of the proposed
DVM algorithm. The proposed efficient DVM algorithm is aimed at implementation using analog realizations. For purposes of evaluation, the algorithm can be realized using both digital hardware as well as software defined radio platforms.
}

\newcommand{\mykeywords}{%
Delay Vandermonde matrix, Efficient algorithms, Self-recursive algorithms, Complexity and performance of algorithms, Approximation algorithms, Wireless communications, Beamforming, Software defined radio
}

\date{}

\begin{document}

\printtitle

\section{Introduction}
\label{intro}
The demand for wireless data communication networks having increased capacity and data transfer rates is growing at a tremendous rate.
The wireless networks are on the verge of rolling out their newest generation of mobile networks; the so-called fifth generation (5G) network, which is expected to be the underlying data transfer network of emerging technologies such as the wireless internet of things (IoT), networks on chip, body-area networks and cyberphysical systems (CPS). Exponentially growing demands for capacity require exponentially growing bandwidth for a given SNR level.
Emerging 5G and 6G wireless networks are built on mm-wave (mmW) bands (typically, 20-600~GHz),
which are of sufficiently high frequencies to allow the necessary growth in system bandwidth.
A 5G wireless data connection may operate around 60~GHz (indoors) and operate at about 200-1000~MHz of bandwidth,
which shows 10- to 25-fold growth in capacity~\cite{R8}.
Such levels of growth are typical of a new generation of wireless networks.
In this paper, we address the problem of obtaining a multitude of directional mmW RF beams using digital signal processing. The paper aims to reduce to arithmetic complexity of the beamforming operation that is based on multiplication of a Vandermonde matrix with an input vector obtained from the array of antennas. The proposed multi-beam signal processor is uses a low-complexity algorithm that enables acceptable digital circuit complexity in order to achieve multi-beam RF beams for a base-station.

The paper is organized as follows. Section~\ref{sec:nbeam} introduces the concept of wideband multi-beam beamforming for wireless communications.
Section~\ref{sec:fac} proposes a self-contained factorization for the DVM and an efficient DVM algorithm, while Section~\ref{sec:com} contains the derivation of arithmetic complexity and elaborate numerical results of the proposed DVM algorithm. Section \ref{sec:bound} furnishes the derivation of a theoretical error bound, establishes numerical results, and addresses the numerical stability of the proposed DVM algorithm. In Section \ref{sec:appln}, applications of the DVM matrix factorization in the engineering discipline will be discussed. Finally, Section~\ref{sec:con} concludes the paper.
\begin{figure*}
\center
 \includegraphics[scale=1.1] {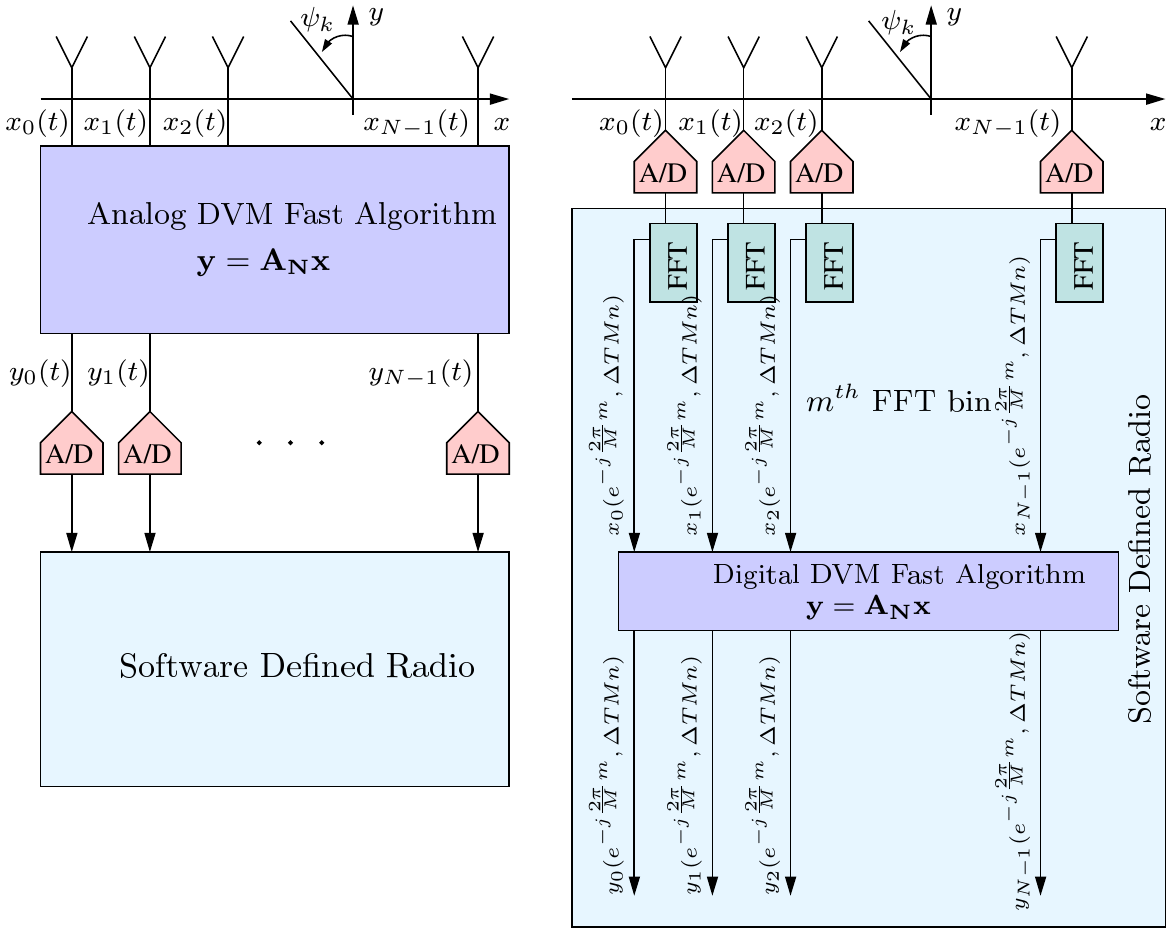}
\caption{A linear array with analog antenna outputs $x_k(t)$ for achieving $N$-beam wideband multi-beam beamformer using left) an analog DVM circuit  that realizes a spatial linear transform ${\bf  y}={\bf A}_N {\bf x}$ employing true time delay analog blocks for RF beams $y_k(t)$; and right) the digital signal flow graph for a parallel digital hardware or software defined radio approach in which an $M$-point discrete Fourier transform (DFT) is computed along the temporal dimension for binning the sampled wideband signals in  temporal frequency domain. Each signal component temporal frequency bin is applied to the DVM algorithm by computing the efficient DVM algorithm for $M$ separate values of $\alpha=e^{-j\omega \tau}$ where $\omega=\frac{2\pi m}{M}.$ In both cases, the DVM algorithm  leads to $N$-beams by coherently combining radio waves in $N$ discrete directions $\psi_k=1,2,\cdots ,N$. Here, temporal sample blocks consist of $M$ samples for the fast Fourier transforms (FFTs). }
\label{F1}
\end{figure*}

\section{Multi-Beam Wideband Array Signal Processing}\label{sec:nbeam}
\subsection {Review of Antenna Beamforming}

The coherent combination of multiple antennas is known as beamforming~\cite{R1, R2}. Fig.\ref{F1} shows an overview of an array processing receiver that operates on an $N$-element array of antennas by applying a spatial linear transform
${\bf A}_N$ across the array outputs $x_k(t), k=1,2, \cdots, N$ in continuous-time to produce a number of continuous-time outputs $y_k(t)$ corresponding to the RF beams.

The linear transform typically takes the form of a spatial discrete Fourier transform, implemented via spatial FFT, and leads to $N$ number of RF beams corresponding to directions pertaining to spatial frequencies of the incident waves that match DFT bin frequencies $2\pi k/N.$  The use of an FFT produces beams that have a frequency dependent axis because it can be shown that the beam orientation is a function of wavelength.

By progressively delaying each antenna by a multiple of a constant time delay,
the RF energy can be directed in a particular direction in a frequency independent manner. For example, let $X_k(e^{j\omega})$ be the Fourier transform of the input $x_k(t)$ for the array outputs $k=1,2, \cdots, N$. The application of a linear delay of duration  $\tau$ causes a corresponding phase rotation by $\omega\tau$. In the frequency domain, the output of the delay becomes $X_k(e^{j\omega})e^{-j\omega\tau}.$ The application of delays
 $\tau kl$ to the $k^{th}$ antenna,  where $l\in\mathbb{Z}$ is an integer causes signal components to be rotated by the frequency dependent phase $\omega \tau kl$. Delay and sum operations - described by a Vandermonde matrix (DVM) by input vector product- leads to $N$ RF beams that have beam orientations $\psi_k,k=1,2, \cdots, N$ that have no dependence on the wavelength.

Example beamformers for arrays can be found in~\cite{SVNA18, R10, R11, R12, R13, R14, R15, R16, R17, R18, VT02}. These systems are \emph{ limited in number of beams, although they do possess relatively high numbers of antennas,} due to the inherent computational/circuit complexity of such multi-beam systems.
Unlike DFT  beamformers,  delay-based Vandermonde matrix multi-beam beamformers have frequency independent beam angles- a property that is desired in wideband basestations.

\subsection{Wideband $N$-Beam Algorithms}
\label{sec:wideband}
For a linear-array with inter-element spacing
$\Delta x$ and speed of light $c$,
the marginal time-difference of arrival between antennas for direction $\psi$ is
$\tau = \frac{\Delta x}{c}\sin\psi$.
The $k$th true-time-delay RF beam can be realized by coherently summing the antenna signals $l=0,1, \cdots, N-1$ such that the beam output is
$y_k(t)=x_0(t)+x_1(t-k\tau)+\cdots+x(t-k(N-1)\tau)$.

Assuming inter-element spacing $\Delta x =\lambda_{min}/2$,
the $N$~simultaneous RF beams in the discrete domain can be written as
the product
of
an $N\times N$ DVM containing the frequency-dependent phase-rotations
and
an input signal vector with elements $x_k(t)$ consisting of the spatial signals from the uniform linear array of antennas.

The wideband multi-beam beamforming algorithm therefore consists of the computation of a DVM-vector product ${\bf y}={\bf A}_N {\bf x}$ at time $t \in \mathbb{R}$ where ${\bf x}$ and ${\bf y}$ are input and output vectors containing signals $x_k(t)$ and $y_k(t)$ respectively.
In our previous work~\cite{SVNA18},
we have proposed a low-complexity DVM algorithm
using the product of complex 1-band upper and lower matrices.
The DVM algorithm in~\cite{SVNA18} extends the results in~\cite{OP00, OA04, Y05}
utilizing complex nodes without considering quasiseparability
and displacement equations as in~\cite{GO941, GO942, P17}.
Moreover,
we have addressed error bounds and stability of the DVM algorithm in~\cite{SVNA18} by filling the gaps in~\cite{OP00, OA04, Y05}.  There are several mathematical techniques available to derive radix-2 and split-radix FFT algorithms,
as described in~\cite{CT65, S86, VL92, JF07, RKH10,blahut2010}. On the other hand, among the known approximation algorithms which run in quadratic arithmetic time to compute Vandermonde matrices by a vector, the author in \cite{P15} presented linear arithmetic time approximation algorithm to compute Vandermonde matrices by vector. But our main intention in this paper is to propose an exact algorithm with self-contained factors not an approximation algorithm.

Even though the derivation of size $N$ DFT into two size $\frac{N}{2}$ DFTs can be done easily,
the extension of this idea to the DVM is cumbersome
as the useful DFT properties are \textbf{not} necessarily present in the DVM case. However, one could still derive an efficient algorithm for the delay Vandermonde matrix as an polynomial evaluation problem.

For wideband analog inputs where $\omega$ spans a continuous range of values, the phase rotations are to be realized using analog delay lines. In analog realizations the DVM fast algorithm is realized in an analog circuit consisting of wideband delays, amplifiers and adders. The low arithmetic complexity of the proposed algorithms will result in correspondingly low circuit complexity for wideband analog multi-beam beamforming circuits.
For purposes of fast verification using software models, the proposed DVM algorithms  assumes a particular  input frequency (i.e., a constant $\omega$) for the incident waves. Such a simplification allows us to compute the relevant matrix-vector product using a computer based numerical simulation model.

During software-based numerical verification, we assume the input signals are over-sampled in time. This is because all computer models much be discrete in time. For sampled digital signal processing systems, the algorithms can be applied at a particular value of $\omega$ in the temporal frequency domain by first computing temporal FFTs for the antenna channels, and then applying the DVM algorithm for each bin of a temporal  $M$-point DFTs.  In such temporally-sampled software/digital implementations, each input $x_k(t)$  becomes $x_k(e^{-j2\pi m/M},t), m=0,1, \cdots,M-1$ and corresponding outputs $y_k(e^{-j2\pi m/M},t) $ where $t=M\Delta Tn$ for temporal sample period $T$ and $M$-point temporal FFT block number $n$. There needs to be $M$ parallel digital circuits, or $M$ calls to the DVM algorithm in software realizations of the fast algorithm  to process all of the temporal frequency bins. There will be call to the DVM algorithm for each temporal FFT output bin at $e^{-j2\pi m/M}$) in order to support wideband operation, and therefore $M$ number of calls to the algorithm in order to process wideband signals transformed by $M$-point temporal FFTs. In our analysis, the code implementations assume $M=1$ because the aim is to numerically model the efficient DVM  algorithm.

 The reported arithmetic complexities scale linearly with $M$ for finer temporal
$M$-point DFTs. For notational simplicity, we will simply use $x_k(t)$ and $y_k(t)$ for describing the efficient DVM algorithm keeping in mind the above mentioned details when considering analog circuit  or software/digital implementations.
\section{Self-Contained Factorizations and fast DVM Beamforming Algorithm}
\label{sec:fac}
The DVM  ${\bf A}_N=[e^{-\tau kl}]_{k=1,l=0}^{N,N-1}$, is a Vandermonde structured matrix with complex entries. Here, we defined $\alpha \equiv e^{-j\omega\tau}$ for notational convenience.  Recall $\tau =\Delta x/c$ is a time delay. On the other hand, the DFT matrix is also a well known Vandermonde structured matrix having $N^{\rm th}$ roots of unity as nodes. In contrast to the DFT, however,
the DVM does not necessarily possess
nice properties,
such as unitary, periodicity, symmetry, and circular shift.

The DVM is defined using distinct complex nodes
$\alpha, \alpha^2, \ldots, \alpha^N$ and hence it is non-singular.
The matrix $\mathbf{A}_N$ can be scaled as
${\bf A}_N=\tilde{\bf A}_N {\bf D}_N$,
where
$\tilde{\bf A}_N=[\alpha^{kl}]_{k,l=0}^{N-1}$
and
${\bf D}_N=\operatorname{diag}[\alpha^{k}]_{k=0}^{N-1}$.
In the following, we will provide a self-contained sparse factorization for $\tilde{\bf A}_N$ followed by the DVM (i.e. ${\bf A}_N$) over complex nodes $\alpha, \alpha^2, \ldots, \alpha^N$.

\begin{lemma}
\label{DVMLem}
Let the scaled delay Vandermonde matrix
$\tilde{\bf A}_{N, \alpha}=[\alpha^{kl}]_{k,l=0}^{N-1}$
be
defined by nodes
$\{1, \alpha, \alpha^2, \ldots, \alpha^{N-1}\} \in \mathbb{C}$
and
$N=2^t$
($t \geq 1$).
Then $\tilde{\bf A}_{N, \alpha}$ can be factored into
\be
\tilde{\bf A}_{N, \alpha}={\bf P}_N^T
\begin{bmatrix}
 \tilde{\bf A}_{\frac{N}{2}, \alpha^2}& \\
& \tilde{\bf A}_{\frac{N}{2}, \alpha^2}
\end{bmatrix}
\left[
\begin{array}{c|c}
 {\bf I}_{\frac{N}{2}}&  {\bf C}_{\frac{N}{2}}^{\frac{N}{2}}\\
\hline\\
\tilde{\bf D}_{\frac{N}{2}} & \alpha^{\frac{N}{2}} {\bf C}_{\frac{N}{2}}^{\frac{N}{2}} \tilde{\bf D}_{\frac{N}{2}}
\end{array}
\right]
,
\label{eq:sDVM}
\ee
where
${\bf P}_N$ is the even-odd permutation matrix, $\tilde{\bf A}_{\frac{N}{2}, \alpha^2}=\left[\alpha^{2kl}\right]_{k,l=0}^{\frac{N}{2}-1}$,
$\bI_{\frac{N}{2}}$ is the identity matrix,
$\tilde{\bD}_{\frac{N}{2}}=\operatorname{diag}[\alpha^{l}]_{l=0}^{\frac{N}{2}-1}$,
and
$\bC_{\frac{N}{2}}$ is the companion matrix defined by the monic polynomial
$p(z)=(z-1)(z-\alpha^2)(z-\alpha^4)\cdots(z-\alpha^{N-2})$
\end{lemma}
\begin{proof}
We show (\ref{eq:sDVM}) by divide-and-conquer technique. We first permute rows of $\tilde{\bf A}_N$ by multiplying with $\bP_N$ and then write the result as the block matrices:
\be
\begin{aligned}
\bP_N\tilde{\bf A}_{N, \alpha}&=\\
&
\left[
\begin{array}{c|c}
\left[\alpha^{2kl}\right]_{k,l=0}^{\frac{N}{2}-1} & \left[\alpha^{2k\left(\frac{N}{2}+l\right)}\right]_{k,l=0}^{\frac{N}{2}-1} \\
\hline
\left[\alpha^{(2k+1)l}\right]_{k,l=0}^{\frac{N}{2}-1}& \left[\alpha^{(2k+1)\left(\frac{N}{2}+l\right)}\right]_{k,l=0}^{\frac{N}{2}-1}
\end{array}
\right]
\end{aligned}
\label{1eq}
\ee
Now, we consider (1,2), (2,1), and (2,2) blocks of $\bP_N\tilde{\bf A}_{N, \alpha}$ (\ref{1eq}) and represent each of these by $\tilde{\bf A}_{\frac{N}{2}, \alpha^2}$ and the product of diagonal matrices.
\\
For (1,2) block of (\ref{1eq}) we get:
\be
\left[\alpha^{2k\left(\frac{N}{2}+l\right)}\right]_{k,l=0}^{\frac{N}{2}-1}={\rm diag}(\alpha^{kN})_{k=0}^{\frac{N}{2}-1} \cdot \begin{bmatrix}
\alpha^{2kl}
\end{bmatrix}_{k,l=0}^{\frac{N}{2}-1}.
\label{1eq12}
\ee
For (2,1) block of (\ref{1eq}) we get:
\be
\left[\alpha^{(2k+1)l}\right]_{k,l=0}^{\frac{N}{2}-1}= \begin{bmatrix}
\alpha^{2kl}
\end{bmatrix}_{k,l=0}^{\frac{N}{2}-1}\cdot {\rm diag}(\alpha^{l})_{l=0}^{\frac{N}{2}-1}.
\label{1eq21}
\ee
For (2,2) block of (\ref{1eq}) we get:
\be
\begin{aligned}
&
\left[\alpha^{(2k+1)\left(\frac{N}{2}+l\right)}\right]_{k,l=0}^{\frac{N}{2}-1}
\\
&
= \alpha^{\frac{N}{2}}
{\rm diag}(\alpha^{kN})_{k=0}^{\frac{N}{2}-1}\begin{bmatrix}
\alpha^{2kl}
\end{bmatrix}_{k,l=0}^{\frac{N}{2}-1}{\rm diag}(\alpha^{l})_{l=0}^{\frac{N}{2}-1}.
\end{aligned}
\label{1eq22}
\ee
Thus by (\ref{1eq12}), (\ref{1eq21}), and (\ref{1eq22}) we can state (\ref{1eq}) as:
\[
\bP_N\tilde{\bf A}_{N, \alpha}=\left[
\begin{array}{c|c}
 \tilde{\bf A}_{\frac{N}{2}, \alpha^2} &  \hat{\bD}_{\frac{N}{2}}\tilde{\bf A}_{\frac{N}{2}, \alpha^2}\\
\hline\\
\tilde{\bf A}_{\frac{N}{2}, \alpha^2} \tilde{\bD}_{\frac{N}{2}} & \alpha^{\frac{N}{2}} \hat{D}_{\frac{N}{2}} \tilde{\bf A}_{\frac{N}{2}, \alpha^2} \tilde{\bD}_{\frac{N}{2}}
\end{array}
\right],
\]
where $\hat{\bD}_{\frac{N}{2}}={\rm diag}(\alpha^{kN})_{k=0}^{\frac{N}{2}-1}$ and  $\tilde{\bD}_{\frac{N}{2}}={\rm diag}(\alpha^{l})_{l=0}^{\frac{N}{2}-1}$.Set $p(z)=(z-1)(z-\alpha^2)(z-\alpha^4)\cdots(z-\alpha^{N-2})=z^{\frac{N}{2}}+\sum_{i=0}^{\frac{N}{2}-1} w_i \cdot z^i$ where $w_i \in \mathbb{C}$. The following equality holds
\be
 \tilde{\bf A}_{\frac{N}{2}, \alpha^2}C_{\frac{N}{2}}=\breve{\bD}_{\frac{N}{2}}\tilde{\bf A}_{\frac{N}{2}, \alpha^2},
\label{deq}
\ee
where
\be
\bC_{\frac{N}{2}}=\begin{bmatrix}
0 & 0 & \cdots & 0 & -w_0\\
 1& 0 & \cdots & 0 &-w_1 \\
 0& 1 & \cdots & 0 &-w_2 \\
 \vdots& \vdots & \ddots & \vdots & \vdots\\
0 & 0 & \cdots & 1 & -w_{\frac{N}{2}-1}
\end{bmatrix}
\label{compM}
\ee
is the companion matrix of the polynomial $p(z)$ with coefficients $w_i (i=0, 1, \cdots, \frac{N}{2}-1)$ and $\breve{\bD}_{\frac{N}{2}}={\rm diag}(\alpha^{2k})_{k=0}^{\frac{N}{2}-1}$. By using (\ref{deq}), non-singularity of $\tilde{\bf A}_{\frac{N}{2}, \alpha^2}$, and induction on $N$, one can easily show
\be
\tilde{\bf A}_{\frac{N}{2}, \alpha^2}\bC_{\frac{N}{2}}^{\frac{N}{2}}=\breve{\bD}_{\frac{N}{2}}^{\frac{N}{2}}\tilde{\bf A}_{\frac{N}{2}, \alpha^2}
\label{Neq}
\ee
for any even number $N$.
Note that, $\breve{\bD}_{\frac{N}{2}}^{\frac{N}{2}}=\hat{\bD}_{\frac{N}{2}}$. Thus
\[
\bP_N\tilde{\bf A}_{N, \alpha}=\left[
\begin{array}{c|c}
 \tilde{\bf A}_{\frac{N}{2}, \alpha^2} &  \tilde{\bf A}_{\frac{N}{2}, \alpha^2}{\bC}_{\frac{N}{2}}^{\frac{N}{2}}\\
\hline\\
\tilde{\bf A}_{\frac{N}{2}, \alpha^2} \tilde{\bD}_{\frac{N}{2}} & \alpha^{\frac{N}{2}, \alpha^2} \tilde{\bf A}_{\frac{N}{2}, \alpha^2}{\bC}_{\frac{N}{2}}^{\frac{N}{2}} \tilde{\bD}_{\frac{N}{2}}
\end{array}
\right]
\]
and we get the result.
\end{proof}

\begin{corollary}
\label{coro:DVM}
Let the delay Vandermonde matrix ${\bf A}_{N, \alpha} = [\alpha^{kl}]_{k=1,l=0}^{N,N-1}$ be defined by nodes $\{\alpha, \alpha^2, \cdots, \alpha^N\}$ and $N=2^t(t \geq 1)$. Then the DVM can be factored into
\be
\begin{aligned}
{\bf A}_{N, \alpha}=\bP_N^T &
\begin{bmatrix}
 {\bf A}_{\frac{N}{2}, \alpha^2}& \\
& {\bf A}_{\frac{N}{2}, \alpha^2}
\end{bmatrix}
\begin{bmatrix}
 {\bar \bD}_{\frac{N}{2}}& \\
& {\bar \bD}_{\frac{N}{2}}
\end{bmatrix} \\
& \left[
\begin{array}{c|c}
 \bI_{\frac{N}{2}} &  \bC_{\frac{N}{2}}^{\frac{N}{2}}\\
\hline\\
\tilde{\bD}_{\frac{N}{2}} & \alpha^{\frac{N}{2}} \bC_{\frac{N}{2}}^{\frac{N}{2}} \tilde{\bD}_{\frac{N}{2}}
\end{array}
\right]{\bf D}_N
\end{aligned}
\label{eq:DVM}
\ee
where ${\bf A}_{\frac{N}{2}, \alpha^2}=[\alpha^{2kl}]_{k=1,l=0}^{\frac{N}{2},\frac{N}{2}-1}$ and ${\bar \bD}_{\frac{N}{2}}={\rm diag}\left[\frac{1}{\alpha^{2k}}\right]_{k=0}^{\frac{N}{2}-1}$.
\end{corollary}
\begin{proof}
This can easily be seen through the scaling of \eqref{eq:sDVM} by ${\bf D}_N$ and ${\bar \bD}_{\frac{N}{2}}$.
\end{proof}
Note that in order to compute the companion matrix
$C_{\frac{N}{2}}$
we have to compute the coefficients of the polynomial
$p(z)=(z-1)(z-\alpha^2)(z-\alpha^4)\cdots(z-\alpha^{N-2})=z^{\frac{N}{2}}+\sum_{i=0}^{\frac{N}{2}-1} w_i \cdot z^i$.
One can do this by setting
$p_{\frac{N}{2}}^{(0)}(z)=1$
and
$p_{\frac{N}{2}}^{(k+1)}=(z-\alpha^{2k})p_{\frac{N}{2}}^{(k)}$
for $k=0, 1, \ldots, \frac{N}{2}-1$.
Then take $p_{\frac{N}{2}}^{(\frac{N}{2})}(z)$ which is $p(z)$. The following lemma gives this procedure.

\begin{lemma}
\label{compcoef}
Let $N$ be an even number, ${ W}=\{1, z^2, z^4, \cdots, z^{N-2}\}$, and $q(z)=\sum_{i=1}^{k} v_i \cdot z^{2i}$, where $k \leq \frac{N}{2}-2$. Then the coefficients of $z^2 \cdot q(z)= \sum_{i=1}^{k+1} w_i \cdot z^{2i}$ can be computed by
\be
\begin{bmatrix}
w_0\\
\vdots\\
w_{k+1}\\
0\\
\vdots\\
0
\end{bmatrix}= \left[\begin{array}{c|c}
{\bf Z} & {\bf O}_{\frac{N}{2}-1}\\
\hline
 \bfe_{\frac{N}{2}-1}& 0
\end{array}\right]\begin{bmatrix}
v_0\\
\vdots\\
v_{k}\\
0\\
\vdots\\
0
\end{bmatrix}
\label{cmt}
\ee
where ${\bf Z}=\begin{bmatrix}
0 & 0 &  \cdots&  0& 0\\
 1& 0 & 0 &  & \vdots\\
0 &  1&  0& \ddots & \vdots\\
 \vdots& \ddots & \ddots & \ddots & 0\\
0 & \cdots & 0 & 1 & 0
\end{bmatrix}$ is the lower shift matrix of size $\left(\frac{N}{2}-1\right) \times \left(\frac{N}{2}-1\right)$, $\bfe_{\frac{N}{2}-1}=\begin{bmatrix}
zeros\left(1,\frac{N}{2}-2\right)& 1
\end{bmatrix}$, ${\bf O}_{\frac{N}{2}-1}=
zeros\left(\frac{N}{2}-1, 1\right)$
\end{lemma}
\begin{proof}
It is obvious that polynomials in ${ W}$ satisfy the recurrence relation $z^{k}=z^2 \cdot z^{k-1}$ for $k=1, 2, \cdots, \frac{N}{2}-1$ with $z^0=1$. By matrix multiplication we can easily get:
\be
\begin{aligned}
z^2 & \begin{bmatrix}
1 & z^2 & z^4 & \cdots &z^{N-2}
\end{bmatrix}  \\
& -\begin{bmatrix}
1 & z^2 & z^4 & \cdots &z^{N-2}
\end{bmatrix}\left[\begin{array}{c|c}
{\bf Z} & {\bf O}_{\frac{N}{2}-1}\\
\hline
 \bfe_{\frac{N}{2}-1}& 0
\end{array}\right] \\
&
=\begin{bmatrix} 0 &\cdots & 0 & z^N\end{bmatrix}
\end{aligned}
\label{cm1}
\ee
Multiplying (\ref{cm1}) by the column of the coefficients we get the result.
\end{proof}
Lemma~\ref{compcoef} can be used to compute the coefficients of the polynomial $p(z)=(z-1)(z-\alpha^2)(z-\alpha^4)\cdots(z-\alpha^{N-2})=z^{\frac{N}{2}}+\sum_{i=0}^{\frac{N}{2}-1} w_i \cdot z^i$ efficiently. Hence the companion matrix $\bC_{\frac{N}{2}}$ can be computed efficiently using the Lemma~\ref{compcoef}.

To compute the self-contained DVM factorization,
first we calculate
the powers of the companion matrix.
We will use the following result for the calculation of
$\bC_{\frac{N}{2}}^{\frac{N}{2}}$.

\begin{corollary}
Let $N=2^{t}(t\geq 2)$, $m=2^k(k \geq 2)$, and $\bC_{\frac{N}{2}}=\begin{bmatrix}
0 & 0 & \cdots & 0 & -w_0\\
 1& 0 & \cdots & 0 &-w_1 \\
 0& 1 & \cdots & 0 &-w_2 \\
 \vdots& \vdots & \ddots & \vdots & \vdots\\
0 & 0 & \cdots & 1 & -w_{\frac{N}{2}-1} \end{bmatrix}$. Then $\bC_{\frac{N}{2}}^{\frac{N}{2}}$ can be computed via
\be
\bC_{\frac{N}{2}}^{m}=\bC_{\frac{N}{2}}^{\frac{m}{2}}\cdot \bC_{\frac{N}{2}}^{\frac{m}{2}},
\label{dccom}
\ee
for $2 \leq m \leq \frac{N}{2}$, where $w_i$ for $i=0, 1, \cdots, \frac{N}{2}-1$ are computed as in Lemma \ref{compcoef}.
\label{poweCom}
\end{corollary}
\begin{proof}
One can easily use induction for $k \geq 2$ to show (\ref{dccom}).
\end{proof}

\begin{remark}
\label{rill}
Although the factorization for the DVM can be stated as in Corollary \ref{coro:DVM}, we should recall here that the classical Vandermonde matrix $V$ is extremely ill-conditioned and in fact the condition number of the matrix $V$ grows exponentially with the size \cite{P16, T94, GI88}. In this paper, we will study how bad the complex structured DVM can be in terms of the choices for nodes in Section \ref{sec:bound}.
\end{remark}

We will first state the following algorithm based on Lemma \ref{compcoef} to compute the coefficients of the polynomial
\be
\begin{aligned}
p(z) &=(z-1)(z-\alpha^2)(z-\alpha^4)\cdots(z-\alpha^{N-2})\\
&=z^{\frac{N}{2}}+\sum_{i=0}^{\frac{N}{2}-1} w_i \cdot z^i.
\end{aligned}
\label{MP}
\ee
Later, the coefficients of $p(z)$ will be used to construct the companion matrix $\bC_{\frac{N}{2}}$ defined in (\ref{compM}).

\begin{algorithm} $({\bf com}(N, \alpha))$\\
\label{algo:compn}
Input:  Even $N$, and ${\bf \alpha} \in \mathbb{C}$
\begin{enumerate}
\item Set $\begin{bmatrix} w_0^{(0)} & w_1^{(0)} &  \cdots & w_{\frac{N}{2}-1}^{(0)}\end{bmatrix}=\begin{bmatrix} 1 & 0 &  \cdots & 0  \end{bmatrix}$
\item For $k=1:N_1-1$,
\[
\begin{aligned}
\begin{bmatrix}
w_0^{(k)}\\
w_1^{(k)}\\
\vdots\\
w_{N_1-1}^{(k)}
\end{bmatrix} & =\left( \left[\begin{array}{c|c}
{\bf Z} & {\bf O}_{N_1-1}\\
\hline
 \bfe_{N_1-1}& 0
\end{array}\right] - \alpha^{2(k-1)} \cdot I \right)
\\
&
\begin{bmatrix}
w_0^{(k-1)}\\
w_1^{(k-1)}\\
\vdots\\
w_{N_1-1}^{(k-1)}
\end{bmatrix}
\end{aligned}
\]
\item Take $\begin{bmatrix} w_0 & w_1 &  \cdots & w_{N_1-1}  \end{bmatrix}=\begin{bmatrix} w_0^{(N_1-1)} & w_1^{(N_1-1)} &  \cdots & w_{N_1-1}^{(N_1-1)}\end{bmatrix}$
\end{enumerate}
Output: Coefficients of $p(z)$ (except the leading coefficient as $p(z)$ is monic) i.e. $\{w_0, w_1, w_2, \cdots, w_{N_1-1}\}$ satisfying \ref{MP}.
\end{algorithm}

We will use the output of algorithm~\ref{algo:compn}
(i.e. ${\bf com}(N, \alpha)$) to construct the companion matrix $\bC_{N_1}$ (\ref{compM}).
Following the self-contained DVM factorization (\ref{eq:DVM}),
one has to compute the powers of
the companion matrix $\bC_{N_1}$ (\ref{compM}).
Corollary~\ref{poweCom}
suggests the following algorithm to compute
$\bC^m_{N_1}$ for $2 \leq m \leq N_1$,
where $m=2^{t_1}(t_1\geq 1)$.

\begin{algorithm} $({\bf comp}(N, \alpha))$\\
\label{algo:comwr}
Input: $N=2^t (\geq 1)$, $N_1=\frac{N}{2}$, and $\alpha \in \mathbb{C}$
\begin{enumerate}
\item Set ${\bf w}=\begin{bmatrix} w_0 & w_1 &  \cdots & w_{N_1-1}\end{bmatrix}$
and $\bC_{N_1}=\left[\begin{array}{c|c}
\begin{bmatrix}
{\bf Z}\\
\bfe_{N_1-1}
\end{bmatrix} & -{\bf w}\\
\end{array}\right]$
\item for $m=2:N_1$\\
 $\bC_{N_1}^m=\bC_{N_1}^{\frac{m}{2}}\bC_{N_1}^{\frac{m}{2}}$\\
 end
\end{enumerate}
Output: $\bC_{N_1}^{N_1}$.
\end{algorithm}

We will use the output of the Algorithm~\ref{algo:comwr} (i.e. ${\bf comp}(N, \alpha)$) to construct $\tilde{\bf C}_N$ s.t. $\tilde{\bf C}_N= \left[
\begin{array}{c|c}
 \bI_{N_1} &  \bC_{N_1}^{N_1}\\
\hline\\
\tilde{\bD}_{N_1} & \alpha^{N_1} \bC_{N_1}^{N_1} \tilde{\bD}_{N_1}
\end{array}
\right]$ for all $N \geq 4$.
The self-contained factorization for the scaled DVM i.e. Lemma \ref{DVMLem} together with algorithms \ref{algo:compn} and \ref{algo:comwr} lead us to establish a recursive radix-2 scaled DVM algorithm to compute $\tilde{\bf A}_{N, \alpha}=[\alpha^{kl}]_{k,l=0}^{N-1}$ as stated next.
\begin{algorithm} $({\bf sdvm}(N, \alpha, {\bf z}))$\\
\label{algo:sdvm}
Input: $N = 2^t (t \geq 1)$, $N_1=\frac{N}{2}$, $\alpha \in \mathbb{C}$, and ${\bf z} \in \mathbb{R}^n {\:\:\rm or\:\:} \mathbb{C}^n$.
\begin{enumerate}
\item Set $\tilde{\bf C}_{N}$.
\item If $N=2$, then \\
\hspace{.1in} ${\bf y}=\begin{bmatrix}
1 & 1\\
1 & \alpha
\end{bmatrix} {\bf z}.$
\item If $N \geq 4$, then \\
 \hspace{.1in} ${\bf u}:=\tilde{\bf C}_N {\bf z}$,\\
\hspace{.1in} ${\bf v1}:={\bf sdvm} \left(N_1, \alpha^2, \left[u_i \right]_{i=0}^{N_1-1} \right)$,\\
 \hspace{.1in} ${\bf v2}:={\bf sdvm} \left(N_1,  \alpha^2, \left[u_i \right]_{i=N_1}^{N}\right)$,\\
 \hspace{.1in} ${\bf y}:={\bf P}_N^T \left({\bf v1}^T, {\bf v2}^T \right)^T$.
\end{enumerate}
Output: ${\bf y}=\tilde{\bf A}_{N, \alpha}{\bf z}$.
\end{algorithm}

\begin{remark}
\label{fcol}
Recall that the delay Vandermonde matrix (i.e. ${\bf A}_{N, \alpha}$) and scaled delay Vandermonde matrix (i.e. $\tilde{\bf A}_{N, \alpha}$) is related via ${\bf A}_{N, \alpha}=\tilde{\bf A}_{N, \alpha}\cdot {\bf D}_N$,
where
${\bf D}_N=\operatorname{diag}(\alpha^k)_{k=0}^{N-1}$.
Thus, once the algorithm~\ref{algo:sdvm} is executed,
we can scale the output of the algorithm by ${\bf D}_N$ to obtain a DVM algorithm.
Although the computational cost of the DVM algorithm reduces in this fashion, the resulting DVM algorithm won't be self-recursive.
\end{remark}

In the following we will state a recursive radix-2 DVM algorithm with the help of Corollary~\ref{coro:DVM}, algorithms~\ref{algo:compn}, and \ref{algo:comwr}.
For notation convenience, we define
$\bar{\bar{\bf \bD}}_N=\begin{bmatrix}
{\bar \bD}_{N_1}& \\
& {\bar \bD}_{N_1}
\end{bmatrix}$ for all $N \geq 4$.

\begin{algorithm} $({\bf dvm}(N, \alpha, {\bf z}))$\\
\label{algo:dvm}
Input: $N = 2^t (t \geq 1)$, $N_1=\frac{N}{2}$, $\alpha \in \mathbb{C}$, and ${\bf z} \in \mathbb{R}^n {\:\:\rm or\:\:} \mathbb{C}^n$.
\begin{enumerate}
\item Set ${\bf D}_N$, $\tilde{\bf C}_{N}$, and $\bar{\bar{\bD}}_N$.
\item If $N=2$, then \\
\hspace{.1in} ${\bf y}=\begin{bmatrix}
1 & \alpha\\
1 & \alpha^2
\end{bmatrix} {\bf z}.$
\item If $N \geq 4$, then \\
 \hspace{.1in} ${\bf u}:={\bf D}_N {\bf z}$,\\
 \hspace{.1in} ${\bf v}:=\tilde{\bf C}_N {\bf u}$,\\
\hspace{.1in} ${\bf r}:=\bar{\bar{\bf D}}_N {\bf v}$,\\
\hspace{.1in} ${\bf s1}:={\bf dvm} \left(N_1, \alpha^2, \left[r_i \right]_{i=0}^{N_1-1}  \right)$,\\
 \hspace{.1in} ${\bf s2}:={\bf dvm} \left(N_1,  \alpha^2, \left[r_i \right]_{i=N_1}^{N} \right)$,\\
 \hspace{.1in} ${\bf y}:={\bf P}_N^T \left({\bf s1}^T, {\bf s2}^T \right)^T$.
\end{enumerate}
Output: ${\bf y}={\bf A}_{N, \alpha}{\bf z}$.
\end{algorithm}

\section{Complexity of DVM Algorithms}
\label{sec:com}
The number of additions and multiplications required to carry out a computation is called the arithmetic complexity.
In this section the arithmetic complexities of the proposed self-recursive scaled DVM and DVM algorithms are established.
\subsection{Arithmetic Complexity of DVM Algorithms}
\label{sub:comtheo}
Here we analyze the arithmetic complexity of the self-recursive scaled DVM and DVM algorithms presented in Section \ref{sec:fac}.
Let
$\#a$
and
$\#m$
denote
the number of complex additions
and
complex multiplications,
respectively,
required to compute
${\bf y}=\tilde{\bf A}_{N, \alpha}{\bf z}$
or
${\bf y}={\bf A}_{N, \alpha}{\bf z}$
for scaled DVM and DVM.
Note that we do not count multiplication by $\pm 1$, $\pm \sqrt{-1}$,
and permutation.
\begin{lemma}
\label{lemma:sdvm}
Let $N=2^t (t \geq 2)$ be given. The arithmetic complexity on computing the scaled DVM algorithm \ref{algo:sdvm} is given by
\begin{align}
\#a(sDVM, N) &= \frac{1}{2}\left(Nt+4^t-N\right),
\nonumber \\
\#m(sDVM, N) &= \frac{3}{2}Nt+\frac{1}{2}4^t-2N.
\label{amdvm1}
\end{align}
\end{lemma}
\begin{proof}
Referring to the ${\bf sdvm}(N, \alpha, {\bf z})$ algorithm, we get
\be
\#a(\textrm{sDVM}, N)  = 2\cdot \#a\left(\textrm{sDVM}, \frac{N}{2} \right) + \#a\left(\tilde{\bf C}_N \right).
\label{aim}
\ee
The matrix $\tilde{\bf C}$ is constructed using $\tilde{\bD}$ and $\bC^{\frac{N}{2}}$. Moreover, to compute the powers of the Companion matrix $\bC^{\frac{N}{2}}$ we have used the divide-and-conquer technique via algorithm ${\bf comp}(N, \alpha)$. Since $m=2^t$ in algorithm ${\bf comp}(N, \alpha)$, by solving a homogeneous first order linear difference equation with respect to $t(t \geq 1)$ (i.e. for $m=2^t (t \geq 1)$ solving $\#a/\#m\left(\bC^m, 2^t \right)-2 \cdot \#a/\#m \left(\bC^m, 2^{t-1}\right)=0$ with initial condition $\#a\left(\bC^m , 2\right)=m-1$ or $\#m\left(\bC^m , 2\right)=m$ respectively), we could obtain $\#a\left(\bC^m \right)=\frac{m^2}{2}-\frac{m}{2}$ and $\#m\left(\bC^{m} \right)=\frac{m^2}{2}$. This fact together with the construction of $\tilde{\bf C}$ using $\tilde{\bD}$ and $\bC^{\frac{N}{2}}$, and $m=\frac{N}{2}$ gives us:
\be
\begin{matrix}
\#a\left(\tilde{\bf C}_N \right)= \frac{N^2}{4}+\frac{N}{2}, &\#m\left(\tilde{\bf C}_N \right) = \frac{N^2}{4}+\frac{3N}{2}
\end{matrix}
\label{Ctn}
\ee
Using the above result we can write (\ref{aim}) as
\[
\#a(\textrm{sDVM}, N)  = 2\cdot \#a\left(\textrm{sDVM}, \frac{N}{2} \right) + \frac{N^2}{4}+\frac{N}{2}
\]
Since $N=2^t$, the above simplifies to the first order difference equation with respect to $t \geq 2$
\[
\#a(\textrm{sDVM}, 2^t)  - 2\cdot \#a\left(\textrm{sDVM}, 2^{t-1}\right) = 4^{t-1}+2^{t-1}.
\]
Solving the above difference equation using the initial condition $\#a(\textrm{sDVM}, 2)=2$, we can obtain
\[
\#a(\textrm{sDVM}, 2^t) =\frac{1}{2}Nt+\frac{1}{2}4^t-\frac{1}{2}N.
\]
Referring the scaled DVM algorithm \ref{algo:sdvm} and  (\ref{Ctn}), we could obtain another first order difference equation with respect to $t \geq 2$
\[
\#m(\textrm{sDVM}, 2^t)  - 2\cdot \#m\left(\textrm{sDVM}, 2^{t-1}\right) = 4^{t-1}+3 \cdot 2^{t-1}.
\]
Solving the above difference equation using the initial condition $\#m(\textrm{sDVM}, 2)=1$, we can obtain
\[
\#m(\textrm{sDVM}, 2^t) =\frac{3}{2}Nt+\frac{1}{2}4^t-2N.
\]
\end{proof}
\begin{lemma}
\label{lemma:dvm}
Let $N=2^t (\geq 2)$, The arithmetic complexity on computing the DVM algorithm \ref{algo:dvm} is given by
\begin{align}
\#a(DVM, N) &=\frac{1}{2}\left(Nt+ 4^t- N\right),
\nonumber \\
\#m(DVM, N) &= \frac{7}{2}Nt+\frac{1}{2}4^t-\frac{7}{2}N.
\label{amdvm2}
\end{align}
\end{lemma}
\begin{proof}
Referring to the ${\bf dvm}(N, \alpha, {\bf z})$ algorithm, we get
\be
\begin{aligned}
\#a(\textrm{DVM}, N)  & = 2\cdot \#a\left(\textrm{DVM}, \frac{N}{2} \right) +  \#a\left({\bf D}_N \right)+
\\ & \hspace{0.2in}
 \#a\left(\tilde{\bf C}_N \right)+  \#a\left(\bar{\bar{\bf D}}_N \right)
\end{aligned}
\label{daim}
\ee
By following the structures of ${\bf D}_N$ and $\bar{\bar{\bf D}}_N$ we get
\be
\begin{matrix}
\#a\left({\bf D}_N \right)= 0, &\#m\left({\bf D}_N \right) = N\\
\#a\left(\bar{\bar{\bf D}}_N \right)= 0, &\#m\left(\bar{\bar{\bf D}}_N \right) = N
\end{matrix}
\label{DDb}
\ee
Thus by using the above and (\ref{Ctn}), we could state (\ref{daim}) as the first order difference equation with respect to $t \geq 1$
 \[
\#a(\textrm{DVM}, 2^t)  - 2\cdot \#a\left(\textrm{DVM}, 2^{t-1}\right) = 4^{t-1}+ 2^{t-1}.
\]
Solving the above difference equation using the initial condition $\#a(\textrm{DVM}, 2)=2$, we can obtain
\[
\#a(\textrm{DVM}, 2^t) =\frac{1}{2}Nt+\frac{1}{2}4^t-\frac{1}{2}N.
\]
Now by using the ${\bf dvm}(N, \alpha, {\bf z})$ algorithm, (\ref{Ctn}), and (\ref{DDb}), we could obtain another first order difference equation with respect to $t \geq 2$
 \[
\#m(\textrm{DVM}, 2^t)  - 2\cdot \#m\left(\textrm{DVM}, 2^{t-1}\right) = 4^{t-1}+ 7\cdot 2^{t-1}.
\]
Solving the above difference equation using the initial condition $\#m(\textrm{DVM}, 2)=2$, we can obtain
\[
\#m(\textrm{DVM}, 2^t) =\frac{7}{2}Nt+\frac{1}{2}4^t-\frac{7}{2}N.
\]
\end{proof}

\subsection{Numerical Results for the Complexity of DVM Algorithms}
\label{sec:numcom}
Numerical results for the arithmetic complexity of the proposed algorithms derived via Lemma \ref{lemma:sdvm} and \ref {lemma:dvm} will be shown in this section.
Figure \ref{figamc} shows the arithmetic complexity of the proposed algorithms vs the direct matrix-vector computations with the matrix size varying from $4 \times 4$ to $4096 \times 4096$. We consider the direct computation of the matrix $\tilde{\bf A}_N$ by the vector ${\bf z}$ cost $N(N-1)$ additions and multiplications (refer to Direct sDVM in Figure \ref{figamc}) and, the matrix ${\bf A}_N$ by the vector ${\bf z}$ cost $N(N-1)$ additions and $N^2$ multiplications (refer to Direct DVM in Figure \ref{figamc}).

\begin{figure*}
\centering
\includegraphics[scale=0.9]{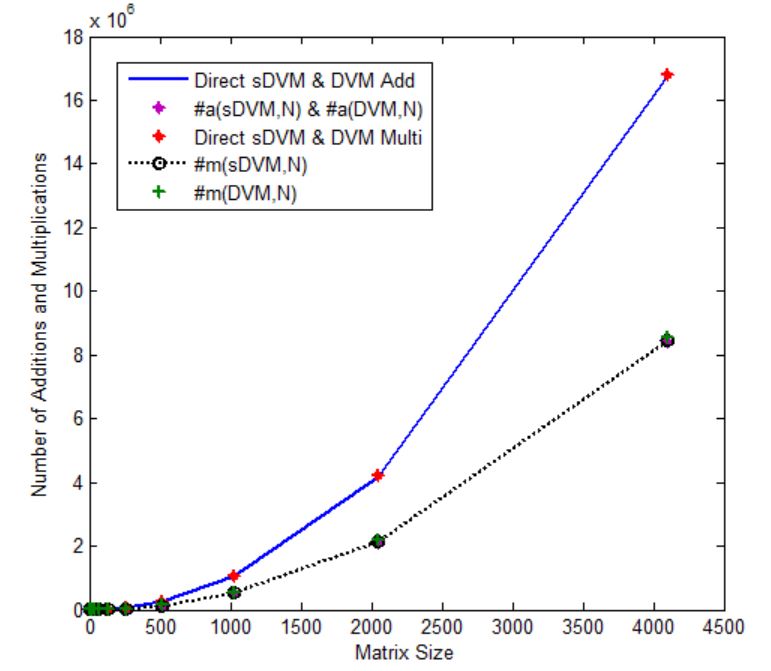}
\caption{Addition and multiplication counts in computing the scaled DVM  and DVM algorithms vs the direct matrix-vector computation.}
\label{figamc}
\end{figure*}

Following the Figure~\ref{figamc}, the scaled DVM and DVM algorithms have the same addition counts and the similar multiplication counts. When the size of the matrices increases the proposed algorithms require fewer addition and multiplication counts as opposed to the direct matrix-vector computation. Moreover, for large $N$, the proposed algorithms have saved $\approx 50\%$ of addition and multiplication counts as opposed to the direct brute-force matrix-vector calculation. As we couldn't distinguish the explicit addition and multiplication counts between the proposed algorithms through the Figure~\ref{figamc}, we have included the explicit counts using the Tables~\ref{tbl:sDVMcost} and~\ref{tbl:DVMcost} in Appendix \ref{appn}. These counts are based on the results obtained in Lemma \ref{lemma:sdvm} and \ref {lemma:dvm}.

\section{Analytic and Numerical Error Bounds of DVM Algorithms}
\label{sec:bound}
\subsection{Theoretical Bounds}
Error bounds of computing the scaled DVM and DVM algorithms is the main concern in this section.
To do so, we use the perturbation of the product of matrices
(stated in~\cite{H96}).
Following the ${\bf sdvm}(N, \alpha, {\bf z})$ and ${\bf dvm}(N, \alpha, {\bf z})$ algorithms,
we have to compute weights $\alpha^{k}$ for $k=0, 1, \ldots, N-1$.
These weights affect the accuracy of the DVM algorithms.
Thus, we will assume that the computed weights
$\widehat{\alpha}^k$ are used and satisfy for all
$k=0, 1, \ldots, N-1$
\be
\widehat{\alpha}^k = \alpha^k + \epsilon_k, \:\:\: |\epsilon_k| \leq \mu,
\label{scerror}
\ee
where $\mu:=cu$,
$u$ is the unit roundoff,
and
$c$ is a constant that depends on the method~\cite{VL92}.

Let's recall the perturbation of the product of matrices
stated in~\cite[Lemma~3.7]{H96}
i.e. if $\bA_k+\Delta \bA_k \in \mathbb{R}^{N \times N}$ satisfies $| \Delta \bA_k | \leq \delta_k |\bA_k|$ for all $k$, then
\be
\begin{matrix}
\Bigg| \displaystyle\prod_{k=0}^m \left( \bA_k+\Delta \bA_k\right) - \displaystyle\prod_{k=0}^m \bA_k   \Bigg|
 \leq \Bigg( \displaystyle\prod_{k=0}^m (1+\delta_k) -1 \Bigg) \displaystyle\prod_{k=0}^m \Bigg| \bA_k \Bigg|
\end{matrix}
,
\nonumber
\ee
where $|\delta_k| < u$.
Moreover, recall $\displaystyle\prod_{k=1}^N (1+\delta_k)^{\pm 1} = 1+\theta_N$ where $|\theta_N| \leq \frac{Nu}{1-Nu}=:\gamma_N$ and $\gamma_k+u \leq \gamma_{k+1}$,  $\gamma_k+\gamma_j+\gamma_k\gamma_j \leq \gamma_{k+j}$
from \cite[Lemma~3.1 and Lemma~3.3]{H96},
and for $x, y \in \mathbb{C}$, $fl(x \pm y)=(x +y)(1 + \delta)$ where $|\delta| \leq u$, $fl(xy)=(xy)(1 + \delta)$ where $|\delta| \leq \sqrt{2}\gamma_2$ from~\cite[Lemma~3.5]{H96}.

In the following, we will prove the error bound on computing the scaled DVM and DVM algorithms.
\begin{theorem}
\label{ThErrsDVM}
Let $\widehat{{\bf y}}=fl(\tilde{\bf A}_N {\bf z})$, where $N=2^t(t \geq 2)$, be computed using the ${\bf sdvm}(N, \alpha, {\bf z})$ algorithm, and assume that (\ref{scerror}) holds. Then
\be
\begin{aligned}
\left | {\bf y}-\widehat{{\bf y}} \right |   \leq   \frac{t \eta}{1-t\eta} & \left| {\bf P}(0) \right|\left|{\bf P}(1)\right|\cdots  \left| {\bf P}(t-2)\right| \left| {\bf \tilde{A}}(t-1) \right| \\
&
\left|{\bf \tilde{C}}(t-2)\right|  \cdots \left| {\bf \tilde{C}}(1)\right| \left| {\bf \tilde{C}}(0)\right| \left|{\bf z}\right|,
\end{aligned}
\label{sEDVM}
\ee
where $\eta=(\mu + \gamma_{2^{t}}(1+\mu))$.
\end{theorem}
\begin{proof}
Using the ${\bf sdvm}(N, \alpha, {\bf z})$ algorithm and the computed matrices $\widehat{{\bf \tilde{C}}}(s)$ at the step numbers (number of executions/iterations of the algorithm) $s=0, 1, 2, \cdots, t-2$ in terms of computed weights $\widehat{\alpha}^k$ for $k=0, 1 \cdots, N-1$, we get
\be
\begin{aligned}
\widehat{{\bf y}}&=
fl \Bigg(
{\bf{P}}(0)\: {\bf{P}}(1)
\cdots
{\bf{P}}(t-2)\:
\widehat{\bf{\tilde{A}}}(t-1)
\\ & \hspace{1in}
\widehat{\bf{\tilde{C}}}(t-2)
\cdots
\widehat{\bf{\tilde{C}}}(2)
\widehat{\bf{\tilde{C}}}(1)
\widehat{\bf{\tilde{C}}}(0)\:
{\bf z}
\Bigg)
\\
&
=
{\bf{P}}(0)\: {\bf{P}}(1)
\cdots
{\bf{P}}(t-2)\:
\left(
\widehat{\bf{\tilde{A}}}(t-1) + \Delta{\widehat{\bf{\tilde{A}}}(t-1)}
\right)
\\
&
\hspace{.3in}
\left(
\widehat{\bf{\tilde{C}}}(t-2) + \Delta{\widehat{\bf{\tilde{C}}}(t-2)}
\right)
\cdots
\left(
\widehat{\bf{\tilde{C}}}(2) + \Delta{\widehat{\bf{\tilde{C}}}(2)}
\right)
\\
&
\hspace{.4in}
\left(
\widehat{\bf{\tilde{C}}}(1) + \Delta{\widehat{\bf{\tilde{C}}}(1)}
\right)
\left(
\widehat{\bf{\tilde{C}}}(0) + \Delta{\widehat{\bf{\tilde{C}}}(0)}
\right)\:
{\bf z},
\end{aligned}
\nonumber
\ee
where ${\bf{P}}(s):=2^s$ block diagonal matrices of $\bP_{2^{t-s}}^T$ and $\widehat{\bf{\tilde{C}}}(s):=2^s$ computed block diagonal matrices of ${\bf{\tilde{C}}}_{2^{t-s}}$. Using the fact that each ${\bf \tilde{C}}(s)$ is computed using the powers of companion matrix $\bC^{\frac{N}{2}}$ with $2^{t-1-s}$ non-zero entry per row, $\bf{\tilde{D}}$ with each one having one non-zero entry per row, and weight ${\alpha}^k$, we get
\be
\begin{matrix}
\left | \Delta{\widehat{\bf{\tilde{C}}}(s)}  \right | \leq {\gamma}_{2^{t-1-s}+3}\left |{\widehat{\bf{\tilde{C}}}(s)}  \right |\:\:\:{\rm for}\:\:\:s=0,1,\ldots,t-2.
\end{matrix}
\label{Ecth1}
\ee
with the use of complex arithmetic. By considering the computed weights $\widehat{\alpha}^k$ and evaluation at the weight ${\alpha}^2$ in each step i.e. using (\ref{scerror});
\be
\begin{matrix}
\widehat{\bf{\tilde{C}}}(s)={\bf{\tilde{C}}}(s)+ \Delta{{\bf{\tilde{C}}}(s)}, \hspace{.1in}  |\Delta{{\bf{\tilde{C}}}(s)}| \leq \mu |{\bf{\tilde{C}}}(s)|.
\end{matrix}
\label{Ecth2}
\ee
Since $\widehat{\bf{\tilde{A}}}(t-1)$ has $2^{t-1}$ block diagonal matrices of $\begin{bmatrix} 1 & 1 \\ 1 & \alpha\end{bmatrix}$, we get
\be
\left | \Delta{\widehat{\bf{\tilde{A}}}(t-1)} \right | \leq {\gamma}_{3}\:\left | {\bf{\tilde{A}}}(t-1) \right |.
\nonumber
\ee
By evaluating $\widehat{\bf{\tilde{A}}}(t-1)$ at the weight ${\alpha}^2$ in each step,
we obtain
\be
\begin{matrix}
\widehat{\bf{\tilde{A}}}(s)={\bf{\tilde{A}}}(s)+ \Delta{{\bf{\tilde{A}}}(s)}, \hspace{.1in}  |\Delta{{\bf{\tilde{A}}}(s)}| \leq \mu |{\bf{\tilde{A}}}(s)|.
\end{matrix}
\nonumber
\ee
Thus overall,
\be
\begin{aligned}
\widehat{{\bf y}}&={\bf{P}}(0)\: {\bf{P}}(1)\cdots{\bf{P}}(t-2)({\bf \tilde{A}}(t-1) + {\bf E}(t-1))
\\& \hspace{.3in}
({\bf \tilde{C}}(t-2)+{\bf E}(t-2))\cdots({\bf \tilde{C}}(1)+{\bf E}(1))
({\bf \tilde{C}}(0)+{\bf E}(0)) {\bf z},
\end{aligned}
\nonumber
\ee
where
$|{\bf E}(s)| \leq (\mu + \gamma_{2^{t}}(1+\mu)) |{\bf \tilde{C}}(s)|$
for $s=0, 1, \ldots, (t-2)$
and
$|{\bf E}(t-1)| \leq (\mu + \gamma_{3}(1+\mu)) |{\bf \tilde{A}}(t-1)|$.
Let $\eta=(\mu + \gamma_{2^{t}}(1+\mu))$.
Hence
\be
\begin{matrix}
\begin{aligned}
\left | {\bf y} -\widehat{{\bf y}} \right | \leq  & \left[ (1+\eta)^{t} -1\right]
 \left| {\bf P}(0) \right|\left|{\bf P}(1)\right|\cdots  \left| {\bf P}(t-2)\right|
\\
& \hspace{.3in} \left| {\bf \tilde{A}}(t-1) \right|\left|{\bf \tilde{C}}(t-2)\right|  \cdots \left| {\bf \tilde{C}}(1)\right| \left| {\bf \tilde{C}}(0)\right| \left|{\bf z}\right|
\\
&  \leq  \frac{t \eta}{1-t\eta}\left| {\bf P}(0) \right|\left|{\bf P}(1)\right|\cdots  \left| {\bf P}(t-2)\right| \left| {\bf \tilde{A}}(t-1) \right|
\\
& \hspace{.3in}
 \left|{\bf \tilde{C}}(t-2)\right|  \cdots \left| {\bf \tilde{C}}(1)\right| \left| {\bf \tilde{C}}(0)\right| \left|{\bf z}\right|
\end{aligned}
\end{matrix}
\nonumber
\ee
Hence the result.
\end{proof}

\begin{theorem}
\label{ErrDVMTh}
Let $\widehat{{\bf y}}=fl({\bf A}_N {\bf z})$,
where $N=2^t$
($t \geq 2$),
be computed using the ${\bf dvm}(N, \alpha, {\bf z})$ algorithm,
and assume that (\ref{scerror}) holds.
Then
\be
\begin{aligned}
\left | {\bf y}-\widehat{{\bf y}} \right | & \leq   \frac{(3t-2) \eta}{1-(3t-2)\eta} \left| {\bf P}(0) \right|\left|{\bf P}(1)\right|\cdots  \left| {\bf P}(t-2)\right|
\\
& \hspace{.3in}
\left| {\bf A}(t-1) \right| \left|{\bf \bar{\bar D}}(t-2)\right|  \cdots \left| {\bf \bar{\bar D}}(1)\right| \left| {\bf \bar{\bar D}}(0)\right|
\\
& \hspace{.4in}
\left|{\bf \tilde{C}}(t-2)\right|  \cdots \left| {\bf \tilde{C}}(1)\right| \left| {\bf \tilde{C}}(0)\right|
\\
& \hspace{.5in} \left|{\bf D}(t-2)\right|  \cdots \left| {\bf D}(1)\right| \left| {\bf D}(0)\right|
\left|{\bf z}\right|
\end{aligned}
\label{EDVM}
\ee
where $\eta=(\mu + \gamma_{2^{t}}(1+\mu))$.
\end{theorem}

\begin{proof}
Using  the ${\bf dvm}(N, \alpha, {\bf z})$  algorithm and the computed matrices $\widehat{\bf {\bar{\bar D}}}(s)$, $\widehat{{\bf \tilde{C}}}(s)$, and $\widehat{\bf{D}}(s)$ at the step numbers (execution/iteration step of the algorithm) $s=0, 1, \ldots, t-2$ in terms of computed weights
$\widehat{\alpha}^k$ for $k=0, 1 \ldots, N-1$, we get
\be
\begin{aligned}
\widehat{{\bf y}}&=fl \Bigg({\bf{P}}(0)\: {\bf{P}}(1)\cdots{\bf{P}}(t-2)\:\widehat{\bf{A}}(t-1)
\\
& \hspace{.5in}
\widehat{\bf {\bar{\bar D}}}(t-2)\cdots \widehat{\bf {\bar{\bar D}}}(2)\widehat{\bf {\bar{\bar D}}}(1)\widehat{\bf {\bar{\bar D}}}(0)
\\
& \hspace{.5in}
\widehat{\bf{\tilde{C}}}(t-2)\cdots \widehat{\bf{\tilde{C}}}_N(2)\widehat{\bf{\tilde{C}}}(1)\widehat{\bf{\tilde{C}}}(0)
\\
& \hspace{.5in}
\widehat{\bf {D}}(t-2)\cdots \widehat{\bf{D}}_N(2)\widehat{\bf{D}}(1)\widehat{\bf{D}}(0)
\: {\bf z} \Bigg)\\
&={\bf{P}}(0)\: {\bf{P}}(1)\cdots{\bf{P}}(t-2)\:\left(\widehat{\bf{A}}(t-1) + \Delta{\widehat{\bf{A}}(t-1)}\right)
\\& \hspace{.05in}
\left( \widehat{\bf {\bar{\bar D}}}(t-2) + \Delta{\widehat{\bf {\bar{\bar D}}}(t-2)} \right) \cdots \left(\widehat{\bf {\bar{\bar D}}}(2) + \Delta{\widehat{\bf {\bar{\bar D}}}(2)} \right)
\\& \hspace{.1in}
\left(\widehat{\bf {\bar{\bar D}}}(1) + \Delta{\widehat{\bf {\bar{\bar D}}}(1)} \right)
\left( \widehat{\bf {\bar{\bar D}}}(0) + \Delta{\widehat{\bf {\bar{\bar D}}}(0)}\right)
\\& \hspace{.05in}
\left( \widehat{\bf{\tilde{C}}}(t-2) + \Delta{\widehat{\bf{\tilde{C}}}(t-2)} \right) \cdots \left(\widehat{\bf{\tilde{C}}}(2) + \Delta{\widehat{\bf{\tilde{C}}}(2)} \right)
\\& \hspace{.1in}
\left(\widehat{\bf{\tilde{C}}}(1) + \Delta{\widehat{\bf{\tilde{C}}}(1)} \right)
\left( \widehat{\bf{\tilde{C}}}(0) + \Delta{\widehat{\bf{\tilde{C}}}(0)}\right)
\\& \hspace{.05in}
\left( \widehat{\bf{D}}(t-2) + \Delta{\widehat{\bf{D}}(t-2)} \right) \cdots \left(\widehat{\bf{D}}(2) + \Delta{\widehat{\bf{D}}(2)} \right)
\\& \hspace{.1in}
\left(\widehat{\bf{D}}(1) + \Delta{\widehat{\bf{D}}(1)} \right)
\left( \widehat{\bf{D}}(0) + \Delta{\widehat{\bf{D}}(0)}\right) {\bf z}.
\end{aligned}
\nonumber
\ee
where ${\bf P}(s):=2^s$ block diagonal matrices of $\bP_{2^{t-s}}^T$, $\widehat{\bf{{\bar{\bar D}}}}(s):=2^s$ computed block diagonal matrices of ${\bar{\bar \bD}}_{2^{t-s}}$, $\widehat{\bf{\tilde{C}}}(s):=2^s$ computed block diagonal matrices of ${\bf{\tilde{C}}}_{2^{t-s}}$, and $\widehat{\bf{D}}(s):=2^s$ computed block diagonal matrices of $\bD_{2^{t-s}}$. Using the fact that each $\bf{{\bar{\bar D}}}(s)$ and ${\bf D}(s)$ have one non-zero entry per row and following complex arithmetic, we get
\be
\begin{matrix}
\left | \Delta{\widehat{\bf{{\bar{\bar D}}}}(s)}  \right | \leq {\gamma}_{2}\left |\widehat{\bf{{\bar{\bar D}}}}(s)  \right |\:\:\:{\rm for}\:\:\:
s=0,1,\ldots,t-2.
\\
\text{and}
\\
\left | \Delta{\widehat{\bf{D}}(s)}  \right | \leq {\gamma}_{2}\left | \widehat{\bf{D}}(s)  \right |\:\:\:{\rm for}\:\:\:
s=0,1,\ldots,t-2.
\end{matrix}
\nonumber
\ee
By considering the computed weights $\widehat{\alpha}^k$ and evaluation at the weight $\alpha^2$ in each step i.e. using (\ref{scerror}),
we have
\be
\begin{matrix}
\widehat{\bf{{\bar{\bar D}}}}(s)={\bf {\bar{\bar D}}}(s)+ \Delta{{\bf {\bar{\bar D}}}(s)}, \hspace{.1in}  |\Delta{{\bf {\bar{\bar D}}}(s)}| \leq \mu |{\bf {\bar{\bar D}}}(s)|  \\
{\rm and}
\\
\widehat{\bf{D}}(s)={\bf D}(s)+ \Delta{{\bf D}(s)}, \hspace{.1in}  |\Delta{{\bf D}(s)}| \leq \mu |{\bf D}(s)| .
\end{matrix}
\nonumber
\ee
Since $\widehat{\bf A}(t-1)$ has $2^{t-1}$ block diagonal matrices of $\begin{bmatrix} 1 & \alpha \\ 1 & \alpha^2\end{bmatrix}$, we get
\be
\left | \Delta{\widehat{\bf{A}}(t-1)} \right | \leq {\gamma}_{4}\:\left| {\bf A}(t-1) \right|.
\nonumber
\ee
By evaluating $\widehat{\bf{A}}(t-1)$ at the weight $\alpha^2$ in each step,
we obtain
\be
\begin{matrix}
\widehat{\bf{A}}(s)={\bf{A}}(s)+ \Delta{{\bf{A}}(s)}, \hspace{.1in}  |\Delta{{\bf{A}}(s)}| \leq \mu |{\bf{\tilde{A}}}(s)|.
\end{matrix}
\nonumber
\ee
Together with (\ref{Ecth1}) and (\ref{Ecth2}) and overall,
\be
\begin{aligned}
\widehat{{\bf y}}&={\bf{P}}(0)\: {\bf{P}}(1)\cdots{\bf{P}}(t-2)({\bf A}(t-1) + {\bf E}(t-1))
\\& \hspace{.1in}
({\bf \bar{\bar D}}(t-2)+{\bf E_1}(t-2))\cdots({\bf \bar{\bar D}}(0)+{\bf E_1}(0))
\\& \hspace{.2in}
({\bf \tilde{C}}(t-2)+{\bf E_2}(t-2))\cdots({\bf \tilde{C}}(0)+{\bf E_2}(0))
\\& \hspace{.3in}
({\bf D}(t-2)+{\bf E_3}(t-2))\cdots({\bf D}(0)+{\bf E_3}(0)) {\bf z},
\end{aligned}
\nonumber
\ee
where $|{\bf E_1}(s)| \leq (\mu + \gamma_{2}(1+\mu)) |{\bf \bar{\bar D}}(s)|$, $|{\bf E_2}(s)| \leq (\mu + \gamma_{2^{t}}(1+\mu)) |{\bf \tilde{C}}(s)|$, and $|{\bf E_3}(s)| \leq (\mu + \gamma_{2}(1+\mu)) |{\bf D}(s)|$ for $s=0, 1, 2, \cdots (t-2)$ and $|{\bf E}(t-1)| \leq (\mu + \gamma_{4}(1+\mu)) |{\bf A}(t-1)|$. Let $\eta=(\mu + \gamma_{2^{t}}(1+\mu))$ and $\eta_1=(\mu + \gamma_{4}(1+\mu))$. Hence
\be
\begin{matrix}
\begin{aligned}
\left | {\bf y} -\widehat{{\bf y}} \right | \leq  & \left[ (1+\eta)^{t-1}(1+\eta_1)^{2t-1} -1\right]
 \left| {\bf P}(0) \right|\cdots  \left| {\bf P}(t-2)\right|
\\
& \hspace{.1in} \left| {\bf A}(t-1) \right| \left|{\bf \bar{\bar D}}(t-2)\right|  \cdots \left| {\bf \bar{\bar D}}(1)\right| \left| {\bf \bar{\bar D}}(0)\right|
\\
& \hspace{.2in}
\left|{\bf \tilde{C}}(t-2)\right|  \cdots \left| {\bf \tilde{C}}(1)\right| \left| {\bf \tilde{C}}(0)\right|
\\
& \hspace{.3in} \left|{\bf D}(t-2)\right|  \cdots \left| {\bf D}(1)\right| \left| {\bf D}(0)\right|
\left|{\bf z}\right|
\\
&  \leq  \frac{(3t-2) \eta}{1-(3t-2)\eta} \left| {\bf P}(0) \right|\left|{\bf P}(1)\right|\cdots  \left| {\bf P}(t-2)\right|
\\
& \hspace{.1in} \left| {\bf A}(t-1) \right| \left|{\bf \bar{\bar D}}(t-2)\right|  \cdots \left| {\bf \bar{\bar D}}(1)\right| \left| {\bf \bar{\bar D}}(0)\right|
\\
& \hspace{.2in}
\left|{\bf \tilde{C}}(t-2)\right|  \cdots \left| {\bf \tilde{C}}(1)\right| \left| {\bf \tilde{C}}(0)\right|
\\
& \hspace{.3in} \left|{\bf D}(t-2)\right|  \cdots \left| {\bf D}(1)\right| \left| {\bf D}(0)\right|
\left|{\bf z}\right|
\end{aligned}
\end{matrix}
\nonumber
\ee
Hence the result.
\end{proof}

Lemma \ref{ThErrsDVM} shows that the forward error bound of the proposed scaled DVM algorithm depends on the size of the matrices $N$, norms of the matrices ${\bf \tilde{C}}(k)$ for $k=0, 1, \cdots, t-2$, and the computed weights. Also, Lemma \ref{ErrDVMTh} shows that the forward error bound of the proposed DVM algorithm depends on the size of the matrices $N$, norms of the matrices ${\bf \bar{\bar D}}(k)$, ${\bf \tilde{C}}(k)$, and ${\bf D}(k)$ for $k=0, 1, \cdots, t-2$, and the computed weights. Thus, the error bound of the proposed scaled DVM and DVM algorithms rapidly increase with the size of the matrices, norms of the powers of matrices, and powers of $\alpha$'s. Hence, the proposed algorithms can not be computed stably for large matrices. This will further be shown through the numerical results in section \ref{sec:num}.
\subsection{Numerical Results}
\label{sec:num}
In this section,
we state numerical results
in connection to the stability of the proposed
algorithm~\ref{algo:dvm} using MATLAB (R2014a version)
with machine precision 2.2204e-16.
Forward error results are presented by taking the exact solutions as the output of the scaled DVM or DVM algorithm computed with the double precision and the computed value as the output of the proposed ${\bf sdvm}(N, \alpha, {\bf z})$ or ${\bf dvm}(N, \alpha, {\bf z})$ algorithms with single precision.
We will show numerical results for matrix sizes from $4 \times 4$ to $128 \times 128$ with $|\alpha| = 1$.

We compare the relative forward error $e$ of the proposed scaled DVM and DVM algorithms defined by
$$
e
=
\frac{\left \| {\bf y}-\hat{\bf y} \right\|_2}
{\left \|{\bf y}\right\|_2}
,
$$
where ${\bf y}={\bf \tilde{A}}_N {\bf z}$ or ${\bf y}={\bf A}_N {\bf z}$ is the exact solution computed using the scaled DVM  or DVM algorithm, respectively,
with double precision and $\hat{\bf y}$ is the computed solution of the algorithms ${\bf sdvm}(N, \alpha, {\bf z})$ or ${\bf dvm}(N, \alpha, {\bf z})$, respectively, with single precision.

Table~\ref{tbl:rande}
shows numerical results for
the forward error of the proposed ${\bf sdvm}(N, \alpha, {\bf z})$
and
${\bf dvm}(N, \alpha, {\bf z})$ algorithms with  $|\alpha|=1$,
and random real and complex inputs
${\bf z_1}$ and ${\bf z_2}$, respectively,
of the scaled DVM and DVM, say Err-sDVM and Err-DVM, respectively.

\begin{table}[h]
\centering
\caption{Forward error in calculating the scaled DVM  and DVM algorithms with $\alpha=e^{-\frac{\pi i}{32}}$, uniformly distributed random input in the interval (0,1), say ${\bf z_1}$ for each $N$, and uniformly distributed random input with real and imaginary parts in the interval (0,1) for each $N$, say ${\bf z_2}$.}
\begin{tabular}{ | c | c | c | c | c | }
    \hline
    {\bf $N$}   & Err-sDVMA & Err-DVM &  Err-sDVM & Err-DVM   \\
   & with ${\bf z_1}$ & with ${\bf z_1}$ & with ${\bf z_2}$ & with ${\bf z_2}$   \\
\hline \hline
4 & 2.367e-08 &  5.648e-08  &   3.577e-08  &  6.855e-08 \\ \hline
8 & 6.118e-08 &  4.952e-08  &  5.959e-08  & 6.820e-08 \\ \hline
16 & 4.676e-08 &  5.529e-08  &  1.010e-08  & 7.449e-08 \\ \hline
32 & 1.022e-08 &  1.262e-07  & 1.067e-08   & 1.279e-07 \\ \hline
64 & 7.008e-08 &   1.138e-07 &  1.568e-07  & 1.373e-07\\ \hline
128 & NaN  &   NaN &  NaN  &NaN\\ \hline
\end{tabular}
\label{tbl:rande}
\end{table}

As shown in Table \ref{tbl:rande}, when $\alpha=e^{-\frac{\pi i}{32}}$ and $N \geq 128$, the MATLAB output will produce NaN for the forward error of the proposed algorithms. This is because the nodes will be repeated and hence the resulting singular matrices while the proposed algorithms are executed for $N > 64$.
Even if $\alpha \neq e^{-\frac{\pi i}{32}}$, but $|\alpha| =1$ and not a root of unity, we have to compute very large powers of matrices
(recall that we compute powers of companion matrices having large powers of $\alpha$'s) and differences of close numbers. Thus, the entries resulting from such operations cannot be represented
as conventional floating-point values and hence lead to undefined numerical values through MATLAB output. This is also evident from the theoretical error bounds obtained in section \ref{sec:bound}.

\section{Future Engineering Tasks}
\label{sec:appln}

\subsection{Analog and/or Digital Circuits that Realize the DVM Algorithm}

Engineering applications require the real-time implementation of the DVM algorithm using a variety of computational platforms. High-speed applications revolving around wireless communications and radar systems typically necessitate analog implementations, which operate on analog signals from an array of sensors, such as antennas. These analog implementations typically employ approximations to ideal time delays in the signal flow graphs, using techniques such as transmission line segments, passive resistor-capacitor lattice filters, or other types of analog delays. Analog realizations, in their most direct form, utilize microwave transmission lines to implement the delays. A microwave transmission line of length $l$ approximates to sufficient accuracy the time delay $T$ where $T=\alpha l/c$ for which $\alpha\le 1$ is the velocity factor of the transmission line. Typically, these transmission lines can be a length of cable of copper track (coplanar waveguide) on a printed circuit board.
When the physical size requirements necessitate smaller circuits, transmission lines can be approximated using analog all-pass filters that can be implemented using integrated circuits \cite{L1, L2}.

Unlike analog DVM circuits requiring delays, digital DVM implementations, may either be in software, using  computer software realizations where the speeds of operation are relatively low (for example, graphics processor units), or in custom digital hardware integrated circuits, for high-speed realizations based on very large scale integration. In both cases, the true time delays found as a basic building block of the DVM algorithm will be approximated using discrete time interpolation filters~\cite{proakis}. For example, various time delays can be rational fractions of the digital systems clock sample period, and can therefore be approximately realized using  both finite impulse response digital interpolation filters as well as infinite impulse response digital interpolation filters. A detailed discussion of the possible approaches for real-time implementation of the DVM algorithms, albeit analog or digital, remains for a future exploration.

\subsection{Low-complexity Algorithms based on Matrix Approximation}

In several applied contexts,
the physics of the problem
admits an appreciable level of error tolerance.
For instance,
this is illustrated in the context of
still image compression~\cite{bas2008},
video encoding~\cite{britanak2007discrete},
beamforming~\cite{suarez2014multi},
motion tracking~\cite{coutinho2017low},
and
biomedical image processing~\cite{cintra2018signal}.
Therefore,
the exact operation
of a given matrix computation
can be relaxed
into
an approximate calculation
that
is carefully tailored to
demand
a
lower arithmetic complexity
when compared
to the original exact computation.
This can be accomplished
by
deriving an approximate matrix
based on the exact matrix.

Approximate matrices can be designed
by several methods,
including
rough inspection,
number representation in dyadic rationals,
and
integer optimization,
to cite a few.
Integer optimization
is often the method of choice
due to its
generality.
The general framework
is described as follows:
\begin{align}
\label{equation-optimization}
\hat{\mathbf{T}}^\ast
=
\arg
\min_{
\hat{\mathbf{T}}
\in
\mathcal{M}_P(N)
}
\operatorname{error}
(
\hat{\mathbf{T}}
,
\mathbf{T}
)
\end{align}
where
$\mathbf{T}$ is the matrix to be approximated
and
$\hat{\mathbf{T}}$
is a candidate matrix
defined
over a low-complexity matrix set~\cite{tablada2015class}.
The error function
is closely linked to the physics of the context
where the approximation is intended to be applied.
Common error functions
are
the Frobenius norm
or
the mean square error~\cite{britanak2007discrete}.
The search space
$\mathcal{M}_P(N)$
is the set of $N\times N$ matrices
with entries defined over the low-complexity integer set~$P$.
A particular common choice for
the set $P$
includes the set of trivial multiplicands
$P_1 = \{0, \pm1, \pm2\}$~\cite{blahut2010}
or
$P_1^2$ for approximations over complex integers~\cite{suarez2014multi}.

For the DVM matrices,
there are two major approaches for deriving
approximations:
(i)~directly approximating the non-factorized
delay Vandermonde matrix
by means of solving~\eqref{equation-optimization}
and
(ii)~approximating only the non-trivial multipliers
in the DVM factorized form
(Corollary~\ref{coro:DVM}).
The former approach has the advantage of
being less restrictive,
but a fast algorithm (factorization)
of the obtained approximation
is left to be derived.
On the other hand,
by approximating from the factorized form,
one has the fast algorithm readily available by construction,
however the derived approximation
is tied to the particular
structure of the considered factorization.
As demonstrated in the context of trigonometric discrete transforms,
approximations
lead to a tunable trade-off between
performance and arithmetic complexity,
often resulting in dramatic reductions in computational cost.
DVM matrices could benefit from similar strategies.

\section{Conclusion}
\label{sec:con}
We have proposed an efficient and self-recursive DVM algorithm having sparse factors. Arithmetic complexities of the proposed algorithm are provided to show that the proposed algorithm is much more efficient than the direct computation of DVM by a vector. The theoretical error bound on computing the proposed algorithm is established. Numerical results of the forward relative error are utilized to analyze the stability of the proposed algorithm. The proposed algorithm lowers the computational complexity of the computation of $N$ parallel RF beams using an array of antennas as detailed in the preceding analysis. Engineering approaches to real-time implementation of the proposed fast algorithms generally take two forms: 1) analog implementations, employing signals which are
continuous in time, continuous in their range, and free of aliasing and quantization effects, and 2) digital implementations, employing signals which are discrete in time (i.e., sampled sequences), discrete in their range (e.g., quantized to be in a set of known values), and therefore susceptible to both aliasing and quantization noise. Typically, discrete domain signals are processed using digital electronics and software, while analog signals are processed using RF integrated circuits, microwave- and mm-wave passive circuits, or photonic integrated circuits. The algorithms proposed here are agnostic to the type of implementation and lend themselves to all types of engineering approaches based on photonics, analog circuits and digital systems, and hybrids thereof. The potential applications of such systems span emerging 5G/6G wireless networks,   wireless IoT, CPS, radio astronomy instrumentation, radar and wireless sensing systems, among others.

\section*{Acknowledgments}
The authors would like to thank Austin Ogle for his valuable time and help that allowed to improve the exposition of the manuscript.

\appendix

\section{Addition and Multiplication Counts}
\label{appn}
The explicit addition and multiplication counts of the proposed scaled DVM  and  DVM  algorithms (proved in Lemma \ref{lemma:sdvm} and \ref {lemma:dvm}) opposed to the direct matrix-vector (which we call as the Direct Add and Direct Multi) computation are shown in Tables \ref{tbl:sDVMcost} and \ref{tbl:DVMcost}).

\begin{table}[h]
\centering
\caption{Arithmetic complexity of the scaled DVM algorithm vs Direct computation}
\begin{tabular}{ | l | l | l | l | l |}
   \hline
    {\bf $N$} & Direct Add & $\#a(sDVM, N)$ & Direct Multi & $\#m(sDVM, N)$ \\
\hline \hline
4&  12  & 10  & 12  & 12   \\ \hline
8& 56 & 40  &  56 & 52 \\ \hline
16  &240  & 152  &  240 & 192  \\ \hline
   32   & 992  &  576 &  992 & 688  \\ \hline
64  & 4032  & 2208  & 4032 & 2496 \\ \hline
128 & 16256 & 8576 & 16256 &9280 \\ \hline
256 &65280 &33664 & 65280 & 35328 \\ \hline
512 & 261632 & 133120 & 261632 & 136960  \\ \hline
1024 & 1047552 & 528896 &  1047552 & 537600 \\ \hline
2048 & 4192256 & 2107392 &  4192256 & 2126848 \\ \hline
4096 & 16773120 & 8411136 &  16773120 & 8454144 \\ \hline
\hline
\end{tabular}
\label{tbl:sDVMcost}
\end{table}
\begin{table}[h]
\centering
\caption{Arithmetic complexity of the DVM algorithm vs Direct computation}
\begin{tabular}{ | l | l | l | l | l |}
    \hline
     {\bf $N$} & Direct Add & $\#a(DVM, N)$ & Direct Multi & $\#m(DVM, N)$ \\
\hline \hline
4&  12  & 10  & 16  &  22  \\ \hline
8& 56 & 40  &   64 & 88 \\ \hline
16  &240  & 152  & 256  &  296 \\ \hline
   32   & 992  &  576 & 1024  &  960  \\ \hline
64  & 4032  & 2208  & 4096 & 3168 \\ \hline
128 & 16256 & 8576 & 16384 & 10880 \\ \hline
256 &65280 &33664 &65536  & 39040 \\ \hline
512 & 261632 & 133120 & 262144 & 145408  \\ \hline
1024 & 1047552 & 528896 &  1048576 & 556544 \\ \hline
2048 & 4192256 & 2107392 &  4194304 &  2168832\\ \hline
4096 & 16773120 & 8411136 &  16777216 & 8546304 \\ \hline
\hline
\end{tabular}
\label{tbl:DVMcost}
\end{table}

\section{Frequently used Abbreviations and Notations}
\label{fab}

\subsection{Abbreviations}
\begin{tabular}{ll}
Cyberphysical systems & CPS\\
Delay Vandermonde matrix & DVM \\
Discrete Fourier transform & DFT\\
Fast Fourier transform & FFT \\
Forward error of the DVM algorithm & Err-DVM\\
Forward error of the scaled DVM algorithm & Err-sDVM\\
Integrated circuit & IC\\
Internet of things & IoT\\
mm wave & mmW \\
multiple-input multiple-output & MIMO\\
radio-frequency & RF \\
Scaled delay Vandermonde matrix & sDVM\\
Signal to noise ratio & SNR\\
\end{tabular}

\subsection{Notations}
\begin{tabular}{ll}
Circular frequency & $\omega$\\
Companion matrix of $p(z)$ & ${\bf C}_{\frac{N}{2}}$ \\
Diagonal matrix & $\tilde{\bD}_{\frac{N}{2}}=\operatorname{diag}[\alpha^{l}]_{l=0}^{\frac{N}{2}-1}$ \\
DVM & ${\bf A}_N:={\bf A}_{N, \alpha}=[\alpha^{kl}]_{k=1,l=0}^{N,N-1}$\\
DVM algorithm & ${\bf dvm}(N, \alpha, {\bf z})$ \\
Node & $\alpha \equiv e^{-j\omega\tau}$\\
Number of additions & \\
in computing  & $\#a(DVM, N)$\\
DVM algorithm & \\
Number of additions & \\
in computing  & $\#a(sDVM, N)$\\
scaled DVM algorithm & \\
Number of multiplications & \\
in computing  & $\#m(DVM, N)$\\
DVM algorithm & \\
Number of multiplications & \\
in computing  & $\#m(sDVM, N)$\\
scaled DVM algorithm & \\
Polynomial with zeros $\alpha^{2k}$ & $p(z)$ \\
Scaled DVM & ${\bf \tilde{A}}_N:={\bf \tilde{A}}_{N, \alpha}=[\alpha^{kl}]_{k,l=0}^{N-1}$\\
Scaled DVM algorithm & ${\bf sdvm}(N, \alpha, {\bf z})$ \\
Speed of light & $c$\\
Temporal frequency & $f$ \\
true-time-delay & $\tau$\\
\end{tabular}

\end{document}